\documentclass[conference,a4paper,10pt]{IEEEtran}
\ifCLASSOPTIONcompsoc
  \usepackage[nocompress]{cite}
\else
  \usepackage{cite}
\fi
\ifCLASSINFOpdf
\else
\fi
\usepackage[utf8]{inputenc}
\usepackage[T1]{fontenc}
\usepackage{lmodern}
\usepackage[scaled]{}

\usepackage{graphicx}  
\usepackage{dcolumn}   
\usepackage{bm}        
\usepackage{amssymb}   
\usepackage{bbm}

\usepackage{amsmath}

\usepackage {amsthm}
\usepackage{enumerate}
\usepackage{chngcntr}
\usepackage{braket}

\newtheorem{theorem}{Theorem}

\newcommand{\diad}[2]{\left|#1\right>\left<#2\right|}
\newcommand{\proj}[1]{\diad{#1}{#1}}

\theoremstyle{definition}
\newtheorem{definition}{Definition}

\newtheorem{corollary}{Corollary}

\newtheorem{conjecture}{Conjecture}

\theoremstyle{remark}

\newcommand\C{\mathbbm{C}}
\def\IC{{\C}}

\newcommand\D{\mathcal{D}}
\newcommand\B{\mathcal{B}}

\newcommand{\Em}{\mathcal{E}}
\newcommand{\Fm}{\mathcal{F}}

\newcommand{\Hm}{\mathcal{H}}

\newcommand{\I}{I}

\newcommand{\id}{\mathcal{I}}

\def\diag{{\rm diag}}
\def\arg{{\rm arg}}
\def\Tr{{\rm Tr}} \def\tr{{\Tr}}
\def\cW{{\cal W}}

\usepackage{verbatim}
\usepackage{booktabs}
\usepackage{tikz}

\usetikzlibrary{arrows,positioning,fit,shapes,calc}

\newdimen\arrowsize
\pgfarrowsdeclare{squarea}{squarea}
{
  \arrowsize=0.4pt
  \advance\arrowsize by.275\pgflinewidth%
  \pgfarrowsleftextend{+-\arrowsize}
  \advance\arrowsize by.5\pgflinewidth
  \pgfarrowsrightextend{+\arrowsize}
}
{
  \arrowsize=0.4pt
  \advance\arrowsize by.275\pgflinewidth%
  \pgfsetdash{}{+0pt}
  \pgfsetroundjoin
  \pgfpathmoveto{\pgfqpoint{1\arrowsize}{4\arrowsize}}
  \pgfpathlineto{\pgfqpoint{-7\arrowsize}{4\arrowsize}}
  \pgfpathlineto{\pgfqpoint{-7\arrowsize}{-4\arrowsize}}
  \pgfpathlineto{\pgfqpoint{1\arrowsize}{-4\arrowsize}}
  \pgfpathclose
  \pgfusepathqfillstroke
}

\pgfarrowsdeclare{open squarea}{open squarea}
{
  \arrowsize=0.4pt
  \advance\arrowsize by.275\pgflinewidth%
  \pgfarrowsleftextend{+-.5\pgflinewidth}
  \advance\arrowsize by7\arrowsize
  \advance\arrowsize by.5\pgflinewidth
  \pgfarrowsrightextend{+\arrowsize}
}
{
  \arrowsize=0.4pt
  \advance\arrowsize by.275\pgflinewidth%
  \pgfsetdash{}{+0pt}
  \pgfsetroundjoin
  \pgfpathmoveto{\pgfqpoint{8\arrowsize}{4\arrowsize}}
  \pgfpathlineto{\pgfqpoint{0\arrowsize}{4\arrowsize}}
  \pgfpathlineto{\pgfqpoint{0\arrowsize}{-4\arrowsize}}
  \pgfpathlineto{\pgfqpoint{8\arrowsize}{-4\arrowsize}}
  \pgfpathclose
  \pgfusepathqstroke
}

\makeatletter
\newdimen\tempa
\newdimen\tempb
\pgfdeclareshape{diamond in circle}{
\inheritsavedanchors[from=diamond] 
\inheritsavedanchors[from=circle] 
\inheritanchorborder[from=circle]
\inheritanchor[from=circle]{center}
\inheritanchor[from=circle]{radius}
\inheritanchor[from=circle]{north}
\inheritanchor[from=circle]{south}
\inheritanchor[from=circle]{east}
\inheritanchor[from=circle]{west}
\inheritanchor[from=circle]{anchorborder}
  \saveddimen\radius{%
    \pgf@ya=.5\ht\pgfnodeparttextbox%
    \advance\pgf@ya by.5\dp\pgfnodeparttextbox%
    \pgfmathsetlength\pgf@yb{\pgfkeysvalueof{/pgf/inner ysep}}%
    \advance\pgf@ya by\pgf@yb%
    \pgf@xa=.5\wd\pgfnodeparttextbox%
    \pgfmathsetlength\pgf@xb{\pgfkeysvalueof{/pgf/inner xsep}}%
    \advance\pgf@xa by\pgf@xb%
    \pgf@process{\pgfpointnormalised{\pgfqpoint{\pgf@xa}{\pgf@ya}}}%
    \ifdim\pgf@x>\pgf@y%
        \c@pgf@counta=\pgf@x%
        \ifnum\c@pgf@counta=0\relax%
        \else%
          \divide\c@pgf@counta by 255\relax%
          \pgf@xa=16\pgf@xa\relax%
          \divide\pgf@xa by\c@pgf@counta%
          \pgf@xa=16\pgf@xa\relax%
        \fi%
      \else%
        \c@pgf@counta=\pgf@y%
        \ifnum\c@pgf@counta=0\relax%
        \else%
          \divide\c@pgf@counta by 255\relax%
          \pgf@ya=16\pgf@ya\relax%
          \divide\pgf@ya by\c@pgf@counta%
          \pgf@xa=16\pgf@ya\relax%
        \fi%
    \fi%
    \pgf@x=\pgf@xa%
    \pgfmathsetlength{\pgf@xb}{\pgfkeysvalueof{/pgf/minimum width}}%
    \pgfmathsetlength{\pgf@yb}{\pgfkeysvalueof{/pgf/minimum height}}%
    \ifdim\pgf@x<.5\pgf@xb%
        \pgf@x=.5\pgf@xb%
    \fi%
    \ifdim\pgf@x<.5\pgf@yb%
        \pgf@x=.5\pgf@yb%
    \fi%
    \pgfmathsetlength{\pgf@xb}{\pgfkeysvalueof{/pgf/outer xsep}}%
    \pgfmathsetlength{\pgf@yb}{\pgfkeysvalueof{/pgf/outer ysep}}%
    \ifdim\pgf@xb<\pgf@yb%
      \advance\pgf@x by\pgf@yb%
    \else%
      \advance\pgf@x by\pgf@xb%
    \fi%
  }
\backgroundpath{
    \tempa=\radius
    \pgfmathsetlength{\pgf@xb}{\pgfkeysvalueof{/pgf/outer xsep}}%
    \pgfmathsetlength{\pgf@yb}{\pgfkeysvalueof{/pgf/outer ysep}}%
    \ifdim\pgf@xb<\pgf@yb%
      \advance\tempa by-\pgf@yb%
    \else%
      \advance\tempa by-\pgf@xb%
    \fi%
    \pgfpathmoveto{\centerpoint\advance\pgf@x by\radius}%
    \pgfpathlineto{\centerpoint\advance\pgf@y by\radius}%
    \pgfpathlineto{\centerpoint\advance\pgf@x by-\radius}%
    \pgfpathlineto{\centerpoint\advance\pgf@y by-\radius}%
    \pgfpathclose%
  }
\behindbackgroundpath{
    \tempa=\radius%
    \pgfmathsetlength{\pgf@xb}{\pgfkeysvalueof{/pgf/outer xsep}}%
    \pgfmathsetlength{\pgf@yb}{\pgfkeysvalueof{/pgf/outer ysep}}%
    \ifdim\pgf@xb<\pgf@yb%
      \advance\tempa by-\pgf@yb%
    \else%
      \advance\tempa by-\pgf@xb%
    \fi%
    \pgfpathcircle{\centerpoint}{\tempa}%
  }
}
\makeatother

\begin{document}
%
\title{Parallel distinguishability of quantum operations}




%
\author{\IEEEauthorblockN{Runyao Duan\IEEEauthorrefmark{1}\IEEEauthorrefmark{2},
Cheng Guo\IEEEauthorrefmark{1},
Chi-Kwong Li\IEEEauthorrefmark{3}, 
and
Yinan Li\IEEEauthorrefmark{1}}
\IEEEauthorblockA{\IEEEauthorrefmark{1}Centre for Quantum Computation and Intelligent Systems, \\Faculty of Engineering and Information Technology,\\
University of Technology Sydney, New South Wales, Australia 2007}
\IEEEauthorblockA{\IEEEauthorrefmark{2}UTS-AMSS Joint Research Laboratory for Quantum Computation and Quantum Information Processing, \\Academy of Mathematics and Systems Science\\Chinese Academy of Sciences, Beijing 100190, China}
\IEEEauthorblockA{\IEEEauthorrefmark{3}Department of Mathematics, \\College of William and Mary, Virginia, 23187-8795, USA}
}


\maketitle

\begin{abstract}
We find that the perfect distinguishability of two quantum operations by a parallel scheme depends only on an operator subspace generated from their Choi-Kraus operators. We further show that any operator subspace can be obtained from two quantum operations in such a way. This connection enables us to study the parallel distinguishability of operator subspaces directly without explicitly referring to the underlining quantum operations. We obtain a necessary and sufficient condition for the parallel distinguishability of an operator subspace that is either one-dimensional or Hermitian. In both cases the condition is equivalent to the non-existence of positive definite operator in the subspace, and an optimal discrimination protocol is obtained. Finally, we provide more examples to show that the non-existence of positive definite operator is sufficient for many other cases, but in general it is only a necessary condition. 
\end{abstract}


%
\IEEEpeerreviewmaketitle

\section{Introduction}
The distinguishability of quantum operations (or intuitively quantum devices) has received great interest in recent years. Compared with the discrimination of quantum states, which is completely characterized by their orthogonality, the distinguishability of quantum operations is more complicated but interesting. In fact, we can choose arbitrary input states as well as arbitrary schemes when distinguishing them. It has been shown that the use of entanglement can significantly improve the discrimination efficiency \cite{Acin2001,DAriano2001,Ji2006,Sacchi2005,Piani2009}. Meanwhile, it has also been shown that by using a sequential scheme, entanglement is not always necessary when distinguishing unitary operations \cite{Duan2007}. Thus there is an interesting trade-off between the spatial resources (entanglement or circuits) and the temporal resources (running steps or discriminating time) when distinguishing quantum operations. More precisely, we consider two basic strategies, the adaptive strategy and the non-adaptive strategy \cite{Harrow2009}. Adaptive strategies allow us to reuse the outputs of previous uses of the quantum operation when preparing the input to subsequent uses; Non-adaptive strategies require that the inputs to all uses of the given operation are chosen before any of them is applied with possible auxiliary systems.

It is worth noting that a sufficient and necessary condition for the perfect distinguishability of quantum operations has been obtained when general adaptive discrimination strategy is used \cite{Duan2009}. However, in practice available resources for discrimination could be very limited and it is not always possible to use adaptive strategies. For instance,  consider the scenario that Alice and Bob are separated by a long distance and share an unknown quantum channel which needs to be identified. When an adaptive protocol is applied, Bob needs to send the intermediate outputs back to Alice for preparing the next input, which requires more resources and infeasible. Clearly non-adaptive strategy would be more suitable in this situation.

\tikzset{
every transaction/.style = {fill=white!100},
transaction/.style = {rectangle, draw=white!100, ultra thick,
    minimum size = 6mm, every actor role},
every actor role/.style = {},
actor role/.style = {rectangle, draw=black!80, ultra thick,
    minimum size = 6mm, every actor role},
composite actor role/.style = {fill=white!80, actor role},
elementary actor role/.style = {fill=white!100, actor role},
initiator/.style = {-},
executor/.style = {<-},
shenglue/.style = {\cdots},
system/.style = {rectangle, fill=white!100, ultra thick, draw=white!80,
            minimum height=60mm, minimum width=3cm,outer sep=0pt}}

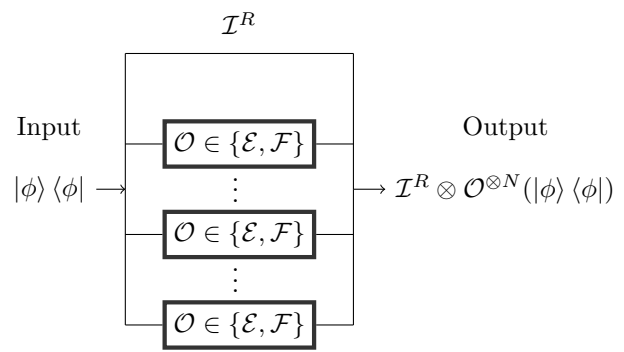
\begin{figure}\label{parallel}
\centering
\pgfdeclarelayer{background}
\pgfdeclarelayer{foreground}
\pgfsetlayers{background,main,foreground}
\begin{tikzpicture}[node distance=1cm, on grid]
    \begin{pgfonlayer}{background}
        \node [system] (system) at (0,3){};
    \end{pgfonlayer}
    \node [transaction] (input) [minimum height=10mm] at ( -2.5,3) {$\proj{\phi}$};
    \node [above] at (input.north) {Input};
    \path ( $(system.south west)!.50!(system.north west)$)
         edge [executor] (input);
     \path ( $(system.south west)!.50!(system.north west)$)
         edge [initiator] ( $(system.south west)!.80!(system.north west)$);
\path ( $(system.south west)!.50!(system.north west)$)
         edge [initiator] ( $(system.south west)!.20!(system.north west)$);
   \path ( $(system.south west)!.80!(system.north west)$)
         edge [initiator] ( $(system.south east)!.80!(system.north east)$);    
\node [transaction] (input) [minimum height=5mm] at ( 0,5.2) {$\id^R$};
\node [composite actor role] (channel1) [minimum width=1.8cm] at ( $(system.south)!.60!(system.north)$){$\mathcal{O}\in \{\Em,\Fm\}$};
\node [composite actor role] (channel2) [minimum width=1.8cm] at ( $(system.south)!.20!(system.north)$){$\mathcal{O}\in \{\Em,\Fm\}$};
\node [composite actor role] (channel3) [minimum width=1.8cm] at ( $(system.south)!.40!(system.north)$){$\mathcal{O}\in \{\Em,\Fm\}$};
\path (channel1) edge [initiator] ( $(system.south west)!.60!(system.north west)$);
\path (channel2) edge [initiator] ( $(system.south west)!.20!(system.north west)$);
\path (channel3) edge [initiator] ( $(system.south west)!.40!(system.north west)$);
\node[text width=3cm] at (1.4,3.1) {$\vdots$};
\node[text width=3cm] at (1.4,1.9) {$\vdots$};
\node [transaction] (output) [minimum height=10mm] at (3.5,3) {$\id^R\otimes\mathcal{O}^{\otimes N}(\proj{\phi})$};
\node [above] at (output.north) {Output};
\path ( $(system.south east)!.60!(system.north east)$) edge [initiator] (channel1);
\path ( $(system.south east)!.20!(system.north east)$) edge [initiator] (channel2);
\path ( $(system.south east)!.40!(system.north east)$) edge [initiator] (channel3);
\path (output) edge [executor] ($(system.south east)!.50! (system.north east)$);
     \path ($(system.south east)!.50! (system.north east)$)
         edge [initiator] ( $(system.south east)!.80!(system.north east)$);
\path ($(system.south east)!.50! (system.north east)$)
         edge [initiator]  ($(system.south east)!.20!(system.north east)$);

\end{tikzpicture}
    \caption{A parallel scheme to distinguish an unknown quantum operation $\mathcal{O}\in\{\Em, \Fm\}$ with $N$ uses, where $\id^R$ represents the identity operator on the auxiliary system $R$.}
\end{figure}

In this paper we focus on non-adaptive strategy, or the parallel scheme, which only allows one to use the unknown operation in parallel. An auxiliary system can also be used if needed, as shown in Figure \ref{parallel}. More precisely, we consider distinguishing two quantum operations $\Em$ and $\Fm$ with respective Choi-Kraus operators $\{E_j: 1\leq j\leq n_0\}$ and $\{F_k:1\leq k\leq n_1\}$ in parallel. We will show that the parallel distinguishability of $\Em$ and $\Fm$  is equivalent to the existence of an integer $N$ such that there is a density operator $\rho\in (S_{\Em,\Fm}^{\otimes N})^\perp$, where 
\begin{equation}\label{subspaceEF}
S_{\Em,\Fm}= {\rm span} \{E_j^\dag F_k:1\leq j\leq n_0, ~1\leq k\leq n_1\}.
\end{equation}
Conversely, for any operator space $T$, we can always find two quantum operations $\Em$ and $\Fm$ such that $T=S_{\Em,\Fm}$. Then we will propose a sufficient and necessary condition for two classes of operator subspaces, namely the one-dimensional operator space and the Hermitian operator space. Specifically, in both of these cases we can always obtain the optimal number of times which we need to tensor. However, this condition is only necessary for general cases.



\section{The characterization of parallel distinguishability}
Consider a $d$-dimensional Hilbert space $\Hm_d$. The set of all linear operators on $\Hm_d$ is denoted by $\B(\Hm_d)$. A general quantum state $\rho$ on $\Hm_d$ is a density operator in $\B(\Hm_d)$ which is positive with trace unity. Moreover, a pure state $\ket{\psi}$ is a unit vector in $\Hm_d$ and the set of all  density operators in $B(\Hm_d)$ is denoted by $\D(\Hm_d)$. Let $\rho$ have the spectral decomposition $\rho=\sum_{k=1}^d\lambda_k\proj{\psi_k}$. The support of $\rho$ is defined by ${\rm supp}(\rho)={\rm span}\{\ket{\psi_k}:\lambda_k>0\}$. Moreover, the Hilbert-Schmidt inner product for $A,B\in\B(\Hm_d)$ is given by $\tr(A^\dag B)$.

A quantum operation $\Em$ from $\B(\Hm_d)$ to $\B(\Hm_{d'})$ is a completely positive and trace-preserving (CPTP) map with the form  $\Em(\rho)=\sum_{i=1}^n E_i\rho E_i^\dag$, where $\{E_i:1\leq i\leq n\}$ are the Choi-Kraus operators of $\Em$ satisfying $\sum_{i=1}^n E_i^\dag E_i=I_d$. An isometry operation is a quantum operation with only one Choi-Kraus operator $U$ such that $U^\dag U=I_d$. Thus, $U$ is a unitary operation when $d=d'$.

\medskip

We first consider the condition under which two quantum operations $\Em$ and $\Fm$ from $\B(\Hm_d)$ to $\B(\Hm_{d'})$ can be distinguished with a single use of the unknown operation. In other words, we want to find a normalized pure input (possibly entangled) state $\ket{\phi}^{RQ}$ such that $(\id_d^R\otimes\Em^Q)(\proj{\phi}^{RQ})$ is orthogonal to $(\id_d^R\otimes\Fm^Q)(\proj{\phi}^{RQ})$, or explicitly,
\begin{equation}\label{orthoEF}
\tr ((\id_d^R\otimes\Em^Q)(\proj{\phi}^{RQ})) ((\id_d^R\otimes\Fm^Q)(\proj{\phi}^{RQ}))=0, 
\end{equation}
where $\Hm_d^R$ denote the auxiliary system and $\Hm_d^Q$ is the principal system under consideration. Choose $\ket{\phi}^{RQ}=(I^R\otimes X^Q)\ket{\Psi}$, where $\Tr(X^{\dagger}X)=1$ and $\ket{\Psi}=\sum_{i=0}^{d-1}\ket{i^R}\ket{i^Q}$. Substituting $\Em$ and $\Fm$ with their Choi-Kraus operators $\{E_{j}: 1\leq j\leq n_0\}$ and $\{F_{k}:{1\leq k\leq n_1}\}$ into Eq. (\ref{orthoEF}), we obtain
$$\sum_{j,k}\Tr(E_{j}^{\dagger}F_{k}XX^{\dagger})\Tr(F_{k}^{\dagger}E_{j}XX^{\dagger})=0,$$
which immediately implies that $XX^\dag$ is orthogonal to $E_j^\dag F_k$ for any $j$ and $k$. Noticing $XX^\dag$ is always positive and trace unity, $XX^\dag\in\D(\Hm_d)$.

When multiple uses of the unknown quantum operation is considered, the calculation is similar and we obtain the following
\begin{theorem}\label{equivalent}
Let $\Em$ and $\Fm$ be two quantum operations from $\B(\Hm_d)$ to
$\B(\Hm_{d'})$ with Choi-Kraus operators $\{E_j:1\leq j\leq n_0\}$ and $\{F_k:1\leq k\leq n_1\}$, respectively. Then they can be perfectly distinguished by $N$ uses in parallel if and only if there is a density operator $\rho\in(S_{\Em,\Fm}^{\otimes N})^\perp$, where $S_{\Em,\Fm}$ is given by Eq. (\ref{subspaceEF}).
\end{theorem}

This theorem shows that the parallel distinguishability is only determined by the operator subspace $S_{\Em,\Fm}$. Moreover, we can easily derive a necessary condition as follows.
\begin{corollary}\label{positive}
If $S_{\Em,\Fm}$ contains a positive definite operator, then $\Em$ and $\Fm$ cannot be perfectly distinguished by a finite number of uses.
\end{corollary}

Meanwhile, we are curious about what kind of operator subspace can be generated by two quantum operations. The following theorem shows that the operator subspace can be chosen freely.


\begin{theorem}\label{free choose}
For any subspace $T\subseteq B(\Hm_d)$, there are two quantum operations $\Em$ and $\Fm$ from $\B(\Hm_d)$ to $\B(\Hm_{d'})$ such that $T=S_{\Em,\Fm}$.
\begin{proof}
We first assume $T$ is spanned by a finite set of operators $\{T_1, T_2,\dots, T_N\}$ where $N\leq d^2$ is the dimension of $T$. Moreover, we assume $T_i^{\dag}T_i\leq I_d$ for $i=1,\dots,N$. We will show that for each $T_i$, there exists two isometries $U_i$ and $V_i$ from $\Hm_d$ to $\Hm_{d'}$ where $d'\geq 2d$ such that $T_i=U_i^\dag V_i$. To see this, let $T_i$ have the singular value decomposition $\sum_{k=1}^{n_i}\sigma_i^k\ket{\psi_i^k}\bra{\phi_i^k}$, where $0\leq \sigma_{i}^{k}\leq 1$.

Define $U_i=\sum_{k=1}^{n_i}\ket{\alpha_i^k}\bra{\psi_i^k}$ and $V_i=\sum_{k=1}^{n_i}\ket{\beta_i^k}\bra{\phi_i^k}$. For each $i$, $\{\ket{\alpha_i^k}:1\leq k\leq n_i\}$ and $\{\ket{\beta_i^k}:1 \leq k\leq n_i\}$ are two sets of orthonormal vectors in $\Hm_{d'}$ to be determined. To make $U_i$ and $V_i$ satisfy $T_i=U_i^\dagger V_i$, we need: 
$$\braket{\alpha_i^j|\beta_i^k}=0~\mbox{and}~\braket{\alpha_i^k|\beta_i^k}=\sigma_i^k,~\mbox{for any}~j,k=1,\dots n_i.$$
This can be done by choosing $n_i$ two-dimensional subspaces in $\Hm_{d'}$ which are mutually orthogonal, and denote them by $K_i^j$ with $j=1,\dots,n_i$. In each $K_i^j$ we choose a basis $\{\ket{\alpha_i^j}, \ket{\beta_{i}^j}\}$ such that $K_i^j={\rm span}\{\ket{\alpha_i^j}, \ket{\beta_{i}^j}\}$ and $\braket{\alpha_i^j|\beta_i^j}=\sigma_i^j$. This can be done since $0\leq \sigma_{i}^j\leq 1$. (In the special case of $\sigma_{i}^j=1$, $K_{i}^j$ is one-dimensional). Note that we can let $d'\geq 2d\geq \max\{2n_i:i=1,2,\dots,N\}$ so such $\{K^j_i\}$ always exist.

Now we construct two quantum operations $\Em$ and $\Fm$ with Choi-Kraus operators $\{E_j:j=1,\dots,N\}$ and $\{F_k:k=1,\dots,N\}$ such that $T=S_{\Em,\Fm}$. This can be done by choosing $E_j=\frac{1}{\sqrt{N}}U_j\otimes \ket{j}$ and $F_k=\frac{1}{\sqrt{N}}V_k\otimes \ket{k}$, where $j,k=1,\cdots, N$.
\end{proof}
\end{theorem}
With the above theorem, we can focus on the operator subspace and we say that an operator subspace $S$ has parallel distinguishability, if there exists a finite positive integer $N$ such that there is a non-zero positive operator in the orthogonal complement of $S^{\otimes N}$. In such a case, we know that for all quantum operations $\Em$ and $\Fm$ such that $S_{\Em,\Fm}=S$, they can be perfectly distinguished by $N$ uses in parallel. 

\section{Parallel distinguishability of two kinds of operator subspaces}
The parallel distinguishability is more difficult than the perfect distinguishability introduced in \cite{Duan2009}, since only limited resources can be used. We want to find some efficient method to check if an operator subspace has parallel distinguishability. By Corollary \ref{positive}, we know that the operator space should not have a positive definite operator. We will introduce two families of operator subspaces such that the parallel distinguishability is only determined by the existence of positive definite operator. 

First, we consider ${\rm dim}(S)=1$. One simple case is $S={\rm span}\{U\}$, where $U$ is a unitary operator. In this situation $\Em$ and $\Fm$ can be chosen as different unitary operations. It is well known that they can be distinguished in parallel \cite{DAriano2001} if and only if $U\neq \I_d$. Let us 
consider $S={\rm span}\{A\}$ where $A\in\B(\Hm_d)$ is not a unitary operator and not positive definite.  We will use the theory of numerical range in our study.
For $A\in\B(\Hm_d)$, let
$$W(A)=\{\bra{\psi}A\ket{\psi}:\ket{\psi}\in\Hm_d, \braket{\psi|\psi}=1\}.$$ 
be the numerical range of $A$, which 
has been researched extensively. It is known that the numerical range of 
an operator $A$ is always convex by the celebrated Toeplitz-Hausdorff
Theorem; for example see  \cite[Chapter 1]{horn2012matrix}.
It is also known that 
the numerical range of a normal operator is just the convex hull 
of its eigenvalues, and  
$W(I\otimes A)=W(A)$. Moreover, we can define the angular numerical range:
\begin{definition}
For a linear operator $A\in\B(\Hm_d)$, the angular numerical range of $A$ is defined as follows:
$$\cW(A) = \cup_{t > 0} W(tA).$$ 
\end{definition}

By the convexity of $W(A)$,  
$\cW(A)$ can be $\IC$, a half space with a straight line passing through 0 as 
the boundary, or a pointed cone with 0 as the vertex.
We can define the field angle of $A$ according to these cases as follows.:

\begin{definition}
For a linear operator $A\in\B(\Hm_d)$, the field angle of $A$, 
denoted by $\Theta(A)$, is defined as follows:
\begin{itemize}
\item[1)] If $\cW(A) = \IC$, $\Theta(A)=2\pi$;
\item[2)] If $\cW(A)$ is a half space, then $\Theta(A) = \pi$;
\item[3)] If $\cW(A)$ is a pointed cone, then $\Theta(A)$ is the angle between 
the two boundary rays of the cone.
\end{itemize}
\end{definition}

\begin{theorem}\label{nonexistenceb}
Consider $S={\rm span}\{A\}$ where $A\in\B(\Hm_d)$. Then $S$ has parallel distinguishability if and only if for any real $t$, $e^{it}A$ is not positive definite. Moreover, there exists a non-zero positive operator in $(S^{\otimes N})^\perp$ if and only if $N\geq \lceil \frac{\pi}{\Theta(A)}\rceil$. 
\end{theorem}

\begin{proof}
Suppose $W(A) \subseteq e^{it}(0, \infty)$ for some real $t$, i.e. $e^{-it}A$ is 
\textbf{positive definite}. Then 
for any positive integer $N$, $W((e^{-it}A)^{\otimes N}) \subseteq (0, \infty)$.  
Thus $S$ does not have parallel distinguishability.

If $0\in W(A)$, i.e. there is $\ket{\psi}$ such that $\tr(A\proj{\psi})=0$. 
In this situation $\proj{\psi}\in S^\perp$ and we are done. 

If $0\not\in W(A)$ and 
$W(A) \not\subseteq e^{it}(0, \infty)$ for any real $t$, then 
there exists a cone in $\IC$ with vertex 0 containing $W(A)$. So, there are
$\mu_1 = r_1 e^{i\theta_1},   
\mu_2 = r_2 e^{i\theta_2} \in W(A)$ with $r_1, r_2 > 0$ and
$\theta_1 < \theta_2 < \theta_1 + \pi$ so that
$\theta_1 \le \arg(\mu) \le \theta_2$ for all $\mu \in W(A)$. Notice that $\Theta(A)=\theta_2-\theta_1$.
We may replace $A$ by $e^{-i\frac{\theta_1+\theta_2}{2}}A$ and assume that 
$$W(A) \subseteq \{\mu \in \IC:  -\frac{\Theta(A)}{2} \le \arg(\mu) \le \frac{\Theta(A)}{2}\}.$$
Let $A = H+iG$, where $H$ and $G$ are Hermitian. Then $H$ is positive definite.
Suppose $U \in \B(\Hm_d)$ is unitary such that
\begin{equation*}
\begin{split}
A_0  &= U^{\dag} H^{-1/2}AH^{-1/2}U \\
     &= U^{\dag}(I_d + iH^{-1/2}GH^{-1/2})U \\
     &=\diag(1+a_1i, \dots, 1+a_di)
\end{split}
\end{equation*}
with $a_1  \ge \dots \ge a_d$. Then $\cW(A)=\cW(A_0)$ and
$a_1 = \tan \frac{\Theta(A)}{2}$ and $a_d = -\tan\frac{\Theta(A)}{2}$.
Furthermore, 
$$\cW(A^{\otimes N}) = \cW(A_0^{\otimes N})=\cW(D^{\otimes N}), $$
where $D = \diag(1 + a_1i, 1+a_di)$.
Hence $0 \in \cW(A^{\otimes N})$ if and only if 
$0 \in \cW(D^{\otimes N})$. Therefore there is a non-zero positive operator in 
$\{A^{\otimes N}\}^\perp$ if and only if $N\ge \frac{\pi}{\Theta(A)}$, which is always finite. 
Choose $N=\lceil {\frac{\pi}{\Theta(A)}}\rceil$ and this will be the smallest 
$N$ such that there is a non-zero positive operator in $S^{\otimes N}$. 
Since for any positive integer $K$ smaller than $N$, 
 $0\not\in W(D^{\otimes K})$. Thus there is no non-zero positive operator in $S^{\otimes K}$.
\end{proof}


Another family of operator subspace, the operator space spanned by a set of Hermitian operators, which has parallel distinguishability if there is no positive definite operator in this space. We have the following:
\begin{theorem}
For an operator subspace $S$ such that $S=S^\dag$ where $S^\dag=\{E^\dag:E\in S\}$, $S$ has parallel distinguishability if and only if there is no positive definite operator in $S$.
\end{theorem}
\begin{proof}
Assume $\{A_1,\dots,A_{N}\}$ is a set of Hermitian operators such that $S={\rm span}\{A_1,\dots,A_{N}\}$. By Farkas' lemma of semi-definite programming \cite{rockafellar2015convex}, either
\begin{itemize}
\item There is a linear combination of $A_1,\dots,A_{N}$ equal to a positive definite operator; or
\item There is a non-zero positive operator $\rho$ such that $\tr(A_i\rho)=0$ for $i=1,\dots,N$.
\end{itemize}

The first statement is equivalent to the existence of a positive definite operator in $S$ and the second one is equivalent to $\rho\in S^\perp$. Thus if there is no \textbf{positive definite} operator in $S$, we can always find a non-zero positive operator in $S^\perp$.  
\end{proof}

\section{Non-existence of a positive definite operator is not always sufficient}
For arbitrary operator subspace, we are curious about if the non-existence of positive definite operator is sufficient for checking the parallel distinguishability. Unfortunately, there exists operator subspace $S$ such that there is no positive definite operator in $S$ but $S$ does not have parallel distinguishability. 
\begin{theorem}\label{counterexample}
Let $S={\rm span}\{A_1,A_2\}\in\B(\Hm_3)$ with $A_1=\proj{0}+i\proj{1}$ and $A_2=\proj{1}+i\proj{2}$, where $\{\ket{0},\ket{1},\ket{2}\}$ is an orthonormal basis of $\Hm_3$. Then for any positive integer $N$, there is no non-zero positive operator in the orthogonal complement of $S^{\otimes N}$.
\end{theorem}
\begin{proof}
It is easy to verify that there is no positive definite operator in $S$, we will show that for arbitrary integer $N$, there is no density operator in $(S^{\otimes N})^\perp$ by mathematical induction.
When $N=1$, by simple calculation there is no density operator in  $S^\perp$.
Assume for $N=k$, there is no density operator in $(S^{\otimes k})^\perp$. Now consider $N=k+1$, where $S^{\otimes k+1}={\rm span}\{A_i\otimes M:i=1,2, M\in S^{\otimes k} \}$. If there exists a
density operator $\rho\in(S^{\otimes k+1})^{\perp}$, we have
\begin{equation*} 
\label{A_1,A_2}
\tr(\rho(A_1\otimes M))=0,\quad \Tr(\rho(A_2\otimes M))=0,
\end{equation*}
for all $M\in S^{\otimes k}$. 
Since $A_1$ and $A_2$ are diagonal, we may assume $\rho=\proj{0}\otimes \rho_0+\proj{1}\otimes \rho_1+\proj{2}\otimes \rho_2$ where $\rho_0$, 
$\rho_1$ and $\rho_2$ are positive operators and at least one of them is non-zero. By substitution we have:
$$
\Tr(\rho_0 M)+i\Tr(\rho_1 M)=0, \quad
-i\Tr(\rho_1M)+\Tr(\rho_2 M)=0.
$$ 
let $\sigma=(\rho_0+\rho_2)/\Tr (\rho_0+\rho_2)$ if $\rho_0+\rho_2\neq 0$, then $\Tr(\sigma M)=0$. Thus $\sigma$ is a density operator in $(S^{\otimes k})^{\perp}$, which is a contradiction.
Thus $\rho_0 = \rho_2 = 0$. But then, we can conclude that $\rho_1$ is non-zero and
$\Tr (\rho_1 M) = 0$ for all $M \in S^{\otimes k}$, again a contradiction.  
\end{proof}

In general, the parallel distinguishability of an operator subspace can be checked by the following system of equations:
\begin{eqnarray}\label{constraint}
    \tr(\rho A_{j_1}\otimes\cdots\otimes A_{j_N})=0,\quad \tr(\rho)=1, \quad \rho\geq 0, \nonumber
\end{eqnarray}
where $A_{j_k}\in S$ for $k=1,\dots,N$ for arbitrary $N$. Notice that this is a semi-definite programming (SDP) problem. However, the size of the SDP will be exponential large and hard to solve. In fact, it is still not easy to check even if we can reduce the problem to a linear programming problem with polynomial size for some special class of operator spaces as shown in the following example. Let 
$$S_\alpha={\rm span}\{A_1=\proj{0}+e^{i\alpha}\proj{1},A_2=\proj{1}+e^{i\alpha}\proj{2}\},$$
where $\alpha\in[\pi/2,\pi]$ guarantees there is no positive definite operator in $S$ and  $\{\ket{0},\ket{1},\ket{2}\}$ is an orthonormal basis of $\Hm_3$.
Notice that $S_\alpha$ is just an operator space spanned by two $3\times 3$ diagonal operators, One can show that there is a non-zero positive operator in the orthogonal 
complement of $S_\alpha^{\otimes N}$ if and 
only if there is a nonzero diagonal
positive operator in the orthogonal 
complement of $S_\alpha^{\otimes N}$.
Thus, we can focus on the search 
of a diagonal positive operator and this can reduce the SDP to a linear programming problem. By simple calculation we can obtain that there is a non-zero diagonal positive operator in $S_\alpha^\perp$ if and only if $\alpha=\pi$. Now consider $S_\alpha^{\otimes 2}$, we need to solve the following linear equation:

$$(A_\alpha\otimes A_\alpha )x = 0,$$

where $A_\alpha = \begin{pmatrix} 1 & e^{i\alpha} & 0\\ 0 & 1 & e^{i\alpha}\end{pmatrix}$.
The above system of linear equations has a nonzero nonnegative solution if and only if $\alpha\in[\frac{3\pi}{4},\pi]$, and we can choose:
$$x=(1, -\cos\alpha, \cos 2\alpha, -\cos\alpha, 
 1, -\cos\alpha,\cos2\alpha, -\cos\alpha, 1).$$

For $N>2$, the number of equations is $2^N$ and the number of variables is $3^N$, which make the problem more challenging. We are going to simplify the problem by using the symmetry. The basic idea to decrease the number of variables is to classify them into different types and assume the variables in each type are same. More precisely, let us consider a fixed $N$ and represent the index of variables $x_0,\dots,x_{3^N-1}$ by ternary numbers ${\underbrace{0\dots0}_{N}},\dots,{\underbrace{2\dots2}_{N}}$. We denote the ternary expansion of $k$ by $\overline{k}$ and the numbers of $0$, $1$ and $2$ in $\overline{k}$ as $\overline{k}_0$, $\overline{k}_1$ and $\overline{k}_2$. In fact, we assume all the index with same number of $0$, $1$ and $2$ are equal, i.e. $x_k=x_{k'}$ if $\overline{k}_i=\overline{k'}_i$ for $i=0,1,2$. This can be done since $S^{\otimes N}$ is invariant under any permutation with order $N$ on its subsystems. Thus if it has a solution it must have a symmetric solution which is invariant under permutations of order $N$. Moreover, we assume:
\begin{equation}\label{partition}
x_k=\frac{\overline{k}_0!\overline{k}_1!\overline{k}_2!}{N!}p_{\overline{k}_0,\overline{k}_1,\overline{k}_2}.
\end{equation}
By this substitution we can reduce the number of variables to $O(N^2)$. Moreover, we can also use this symmetry to reduce the number of equations. Let $\rho={\rm diag}(x_0,\dots,x_{3^N})$ be the desired diagonal positive operator. Considering the first equation: $\Tr (A_1^{\otimes N}\rho)=0$. it is easy to see the index of variables will be all the binary numbers from $0$ to $2^N-1$. Denote the binary expansion of $k$ by $\hat{k}$ and the number of $1$ in $\hat{k}$ by $\hat{k}_1$ similarly. We can rewrite the equation by:  
\begin{equation}\label{eq2}
\sum_{j=0}^{2^N-1} e^{i\hat{j}_1\alpha}x_j=0
\end{equation}
Moreover, let $\mathcal{J}_r=\{j:\hat{j}_1=r,0\leq j\leq 2^N-1\}$, then $|\mathcal{J}_r|= \binom{N}{r}$ and $\binom{N}{r}x_j=p_{N-r,r,0}$ if $j\in \mathcal{J}_r$. Thus we can rewrite the above equation as:
$$\sum_{r=0}^{N} e^{ir\alpha} p_{N-r,r,0}=0.$$

Moreover, by the symmetry we can ignore the place of $A_2$ when we substitute it with $A_1$. We only need to consider equations $\tr(A_2^{\otimes l}\otimes A_1^{\otimes N-l}\rho)=0$ for $l=0,\dots,N$, which reduce the number of equations to $O(N)$. Furthermore, the substitution of $A_1$ by $A_2$ will permute the index of variables in Equation \ref{eq2} but remain the coefficient unchanged. More precisely, for an index $k$ with ternary expansion $\overline{k}=y_{N-1}3^{N-1}+\cdots+y_13+y_0$, if we substitute the first $l$ $A_1$ by $A_2$, the index after substitution will be $\overline{k}_l=(y_{N-1}+1)3^{N-1}+\cdots+(y_{N-l}+1)3^{N-l}+y_{N-l-1}3^{N-l-1}+\cdots+y_0$. Noticing that we use binary number $\hat{k}=z_{N-1}2^{N-1}+\cdots+z_12+z_0$ in Equation \ref{eq2}, we define the following map $\overline{f}$ to transfer a binary number $\hat{k}$ to a ternary number, i.e. $\overline{f}(\hat{k})=z_{N-1}3^{N-1}+\cdots+z_13+z_0$.  Thus $\tr(A_2^{\otimes l}\otimes A_1^{\otimes N-l}\rho)=0$ can be rewritten as:
\begin{equation}\label{eq3}
\sum_{j=0}^{2^N-1} e^{i\hat{j}_1\alpha}x_{\overline{f}(\hat{j})_l}=0
\end{equation}

Since we are going to use Equation \ref{partition} to represent $x_{\overline{f}(\hat{j})_l}$. We consider the elements in $\mathcal{J}_r$. Denote the number of $0$ in the first $l$ positions of $j\in \mathcal{J}_r$ by $s$, where $0\leq s\leq N-r$. After the substitution we obtain $\binom{l}{s}\binom{N-l}{N-r-s}$ ternary numbers with value of $$\frac{(N-r-s)!(2s+r-l)!(l-s)!}{N!}p_{N-r-s,2s+r-l,l-s}.$$ Thus we rewrite Equation \ref{eq3} as:
$$\sum_{r=0}^{N}\!e^{ir\alpha}\big[\!\sum_{s=0}^{N-r}\binom{2s+r-l}{s}p_{N-r-s,2s+r-l,l-s}\big]\!=\!0.$$
Using the above simplified equations, we can obtain the explicit solution for $N=3,4$.

For $N=3$, we have:
\begin{itemize}
\item $p_{0,3,0}=0$;
\item $p_{3,0,0}=p_{0,0,3}=p_{1,2,0}=p_{0,2,1}=\sin\alpha$;
\item $p_{2,1,0}=p_{0,1,2}=-\sin2\alpha$;
\item $p_{1,1,1}=-2\sin2\alpha$;
\item $p_{1,0,2}=p_{2,0,1}=\sin3\alpha$.
\end{itemize}
When $\alpha\in[\frac{2\pi}{3},\frac{3\pi}{4}]$, the above values are all non-negative. When $\alpha\in[\frac{3\pi}{4},\pi]$, we can simply use the solution for $N=2$ to construct the solution. 

For $N=4$, we have the following  solution:
\begin{itemize}
\item $p_{4,0,0}=p_{0,0,4}=p_{0,4,0}=1$;
\item $p_{0,2,2}=p_{2,2,0}=2$;
\item $p_{0,1,3}=p_{3,1,0}=p_{0,3,1}=p_{1,3,0}=-2\cos\alpha$;
\item $p_{1,1,2}=p_{2,1,1}=2\cos3\alpha$;
\item $p_{2,0,2}=-2\cos4\alpha$;
\item All the rest entries are $0$.
\end{itemize}
When $\alpha\in[\frac{5\pi}{8},\frac{2\pi}{3}]$, the above values are all non-negative.

Up to $N = 18$, we are able to determine
a solution for $\rho$ numerically for 
different choice of $\alpha$. We have 
the following conjecture.
\begin{conjecture}\label{conjecture}
For operator space $S_\alpha=\{\proj{0}+e^{i\alpha}\proj{1}, \proj{1}+e^{i\alpha}\proj{2}\}$ with $\alpha\in[\frac{\pi}{2},\pi]$, there is a density operator in the orthogonal complement of $(S_\alpha)^{\otimes N}$ if and only if $\alpha\in[\frac{\pi}{2}+\frac{\pi}{2N}, \pi]$.
\end{conjecture} 

It is interesting to note that the conjecture is true when $N\to \infty$. Since $S_{\frac{\pi}{2}}$ is the case in Theorem \ref{counterexample}, we have already shown that the desired $N$ does not exist.

\section{Conclusion}
In this paper we discuss the problem of parallel distinguishability of general quantum operations. We show that the parallel distinguishability is determined by an operator subspace generated by their Choi-Kraus operators. Meanwhile, the operator subspace can be chosen arbitrarily. We introduce the parallel distinguishability of an operator subspace and focus on characterizing operator subspaces which have this property. Furthermore, we show the parallel distinguishability of one-dimensional operator spaces and Hermitian operator spaces can be verified by checking if there exists a positive definite operator. However, an example is given to show that this condition is not always sufficient. We also construct a family of operator subspaces and obtain some analytical and numerical results as well as a conjecture about the full characterization of the parallel distinguishability of this family.

RD is supported in part by the Australian Research Council (ARC) under Grant DP120103776 and by the National Natural Science Foundation (NNSF) of China under Grant 61179030; furthermore in part by an ARC Future Fellowship under Grant FT120100449.  CL is supported in part by the USA NSF grant DMS 1331021, the Simons Foundation Grant 351047, and the NNSF of China Grant 11571220. The research began when
he visited the Centre for Quantum Computation and Intelligent Systems
of the University of Technology Sydney
in June 2015. He would like to express
his thanks to the hospitality of the 
colleagues of University of Technology
Sydney.



%

\bibliographystyle{IEEEtranS}
\bibliography{Mycollection}

\begin{thebibliography}{10}
\providecommand{\url}[1]{#1}
\csname url@samestyle\endcsname
\providecommand{\newblock}{\relax}
\providecommand{\bibinfo}[2]{#2}
\providecommand{\BIBentrySTDinterwordspacing}{\spaceskip=0pt\relax}
\providecommand{\BIBentryALTinterwordstretchfactor}{4}
\providecommand{\BIBentryALTinterwordspacing}{\spaceskip=\fontdimen2\font plus
\BIBentryALTinterwordstretchfactor\fontdimen3\font minus
  \fontdimen4\font\relax}
\providecommand{\BIBforeignlanguage}[2]{{%
\expandafter\ifx\csname l@#1\endcsname\relax
\typeout{** WARNING: IEEEtranS.bst: No hyphenation pattern has been}%
\typeout{** loaded for the language `#1'. Using the pattern for}%
\typeout{** the default language instead.}%
\else
\language=\csname l@#1\endcsname
\fi
#2}}
\providecommand{\BIBdecl}{\relax}
\BIBdecl

\bibitem{Acin2001}
A.~Ac{\'{\i}}n, ``{Statistical distinguishability between unitary
  operations.}'' \emph{Physical review letters}, vol.~87, no.~17, p. 177901,
  2001.

\bibitem{DAriano2001}
G.~M. D'Ariano, P.~{Lo Presti}, and M.~G.~A. Paris, ``{Using Entanglement
  Improves the Precision of Quantum Measurements},'' \emph{Physical review
  letters}, vol.~87, no.~27, p. 270404, 2001.

\bibitem{Duan2007}
R.~Duan, Y.~Feng, and M.~Ying, ``{Entanglement is Not Necessary for Perfect
  Discrimination between Unitary Operations},'' \emph{Physical review letters},
  vol.~98, no.~10, p. 100503, 2007.

\bibitem{Duan2009}
R.~Duan, Y.~Feng, and M.~Ying, ``{Perfect Distinguishability of Quantum
  Operations},'' \emph{Physical Review Letters}, vol. 103, no.~21, p. 210501,
  2009.

\bibitem{Harrow2009}
A.~W. Harrow, A.~Hassidim, D.~W. Leung, and J.~Watrous, ``{Adaptive versus
  non-adaptive strategies for quantum channel discrimination},'' \emph{Physical
  Review A}, vol.~81, no.~3, p.~11, 2009.

\bibitem{horn2012matrix}
R.~A. Horn and C.~R. Johnson, \emph{Matrix analysis}.\hskip 1em plus 0.5em
  minus 0.4em\relax Cambridge university press, 2012.

\bibitem{Ji2006}
Z.~Ji, Y.~Feng, R.~Duan, and M.~Ying, ``{Identification and Distance Measures
  of Measurement Apparatus},'' \emph{Physical Review Letters}, vol.~96, no.~20,
  p. 200401, 2006.

\bibitem{Piani2009}
M.~Piani and J.~Watrous, ``{All entangled states are useful for channel
  discrimination},'' \emph{Physical Review Letters}, vol. 102, no. June, pp.
  1--4, 2009.

\bibitem{rockafellar2015convex}
R.~T. Rockafellar, \emph{Convex analysis}.\hskip 1em plus 0.5em minus
  0.4em\relax Princeton university press, 2015.

\bibitem{Sacchi2005}
M.~F. Sacchi, ``{Optimal discrimination of quantum operations},''
  \emph{Physical Review A}, vol.~71, no.~6, p. 062340, 2005.

\end{thebibliography}

\end{document}